\documentclass{article}
\usepackage{spconf,amsmath}

\usepackage{amssymb, amsthm, amscd, MnSymbol,mathrsfs}
\DeclareMathAlphabet{\mathbbm}{U}{bbm}{m}{n}
\usepackage{graphicx}
\usepackage{nicefrac}
\usepackage{bm}

\usepackage{url}

\usepackage{cite}
\usepackage{balance}

\usepackage[caption=false,font=normalsize,labelfont=sf,textfont=sf]{subfig}

\usepackage{algorithmicx}
\usepackage[linesnumberedhidden,ruled,vlined]{algorithm2e}
\usepackage{color}

\usepackage{wrapfig} 
\usepackage{pgfplots}

\usepackage[para,online,flushleft]{threeparttable}
\usepackage{multirow}

\newtheorem{theorem}{Theorem}[section]

\newtheorem{lemma}[theorem]{Lemma}
\newtheorem{proposition}[theorem]{Proposition}
\newtheorem{corollary}[theorem]{Corollary}

\numberwithin{equation}{section}
\mathchardef\mhyphen="2D
\title{Phaseless Super-resolution in the Continuous Domain}
\name{{Myung~Cho$^{1}$, Christos~Thrampoulidis$^{2}$, Weiyu~Xu$^{1}$, and Babak~Hassibi$^{2}$}
}
\address{
\small
$^{1}$ Dept. of ECE, University of Iowa, Iowa City, IA, 52242\\
\small
$^{2}$ Dept. of EE, California Institute of Technology, Pasadena, CA, 91125
\normalsize
}
\begin{document}
%
\maketitle
\begin{abstract}
Phaseless super-resolution refers to the problem of super-resolving a signal from only its low-frequency Fourier magnitude measurements. In this paper, we consider the phaseless super-resolution problem of recovering a sum of sparse Dirac delta functions which can be located anywhere in the continuous time-domain. For such signals in the continuous domain, we propose a novel Semidefinite Programming (SDP) based signal recovery method to achieve the phaseless super-resolution. This work extends the recent work of Jaganathan \textit{et al.} \cite{jaganathan2016phaseless}, which considered phaseless super-resolution for discrete signals on the grid.
\end{abstract}
\section{Introduction}
\label{sec:intro}
In engineering and science, improving the accuracy and precision of measurement tools, such as microscopy, X-ray crystallography and MRI, is of great interest. However, due to the physical limitations in measurement tools, sometimes we can only indirectly or partially observe a signal of interest, e.g., obtaining only low-frequency information, only low-resolution image, or only the magnitude of a signal. The microscope is a good example of a measurement tool having such physical limitations ranging from low-frequency measurements to phaseless measurements \cite{brenner1959negative,zhang2006functional,benveniste2002mr,stolyarova2007high,huang2009super}. 

To overcome the limitation of low-frequency measurements, researchers have investigated recovering a signal from only its low-frequency Fourier measurements, and referred to it as \textit{super-resolution}. The authors in \cite{candes2014towards} and \cite{tang2012csotg} proposed SDP based methods for the recovery of signals in the continuous domain under certain separation conditions, by employing \textit{Total Variation Norm Minimization (TVNM)} and \textit{Atomic Norm Minimization (ANM)} respectively.
Besides, to address the issue of phaseless measurements, people have studied the \textit{phase retrieval} to obtain phase information from the magnitude measurements of a signal \cite{fienup1982phase,gerchberg1972practical}. Recently, in \cite{candes2015phase}, the authors proposed a trace-norm minimization to solve the phase retrieval problem with the use of masks.

Super-resolving a signal from only magnitudes of low-frequency Fourier measurements is often ill-posed due to lack of both phase information and high-frequency information; and hence it is a challenging problem. The authors in \cite{chen2014algorithm,jaganathan2016phaseless} considered the \textit{phaseless super-resolution} aiming at recovering signals with only low-frequency magnitude measurements. In the noiseless setting, the authors in \cite{chen2014algorithm} proposed a combinatorial algorithm for signal recovery using only low-frequency Fourier magnitude measurements, but this algorithm requires additional distinguishing conditions on the signal impulses. In the noisy setting, this combinatorial algorithm suffers from error propagation. Instead of assuming the distinguishing conditions on signals, the authors in \cite{jaganathan2016phaseless} used masks to obtain different types of magnitude measurements. The authors provably showed that under appropriate choice of masks, an SDP formulation can be used to recover time-domain impulse signals on the discretized grid.

In this paper, we consider super-resolving time-domain impulse signals located off the grid from only low-frequency Fourier magnitude measurements. To tackle the continuous parameter domain, we propose a novel SDP formulation, employing ANM to recover signals from Fourier magnitude measurements. For example, our approach applies to the magnitude measurements used in \cite{candes2015phase,jaganathan2016phaseless}. In numerical experiments, we show the successful recovery of signals in the continuous domain from only low-frequency magnitude measurements. Furthermore, we compare our algorithm to a simple combining algorithm performing phase retrieval followed by ANM. Our method shows better recovery performance than the simple combining algorithm. In the future work, we will consider the noisy magnitude measurement case.

\textbf{Notations:} In this paper, we denote the set of complex numbers as $\mathbb{C}$. We reserve calligraphic uppercase letters for index sets, e.g., $\mathcal{N}$. We use $|\mathcal{N}|$ as the cardinality of the index set $\mathcal{N}$. We use the superscripts $*$, $T$, and $H$ to denote conjugate, transpose, and conjugate transpose respectively. We reserve $i$ for the imaginary number, i.e., $i^2 = -1$. We denote a time-domain signal as a lowercase letter, and its frequency-domain signal as its uppercase letter. To denote a ground true signal, we use the superscript $o$, e.g., $x^{o}$. For the index of a vector and a matrix, we start with the index $0$; hence, we denote the first element of the vector $X$ as $X_0$, and the top-left element of a matrix $Q$ as $Q_{0,0}$.

\section{Problem Formulation and Background}
Let $x^{o}(t)$ be a sum of Dirac functions expressed as follows:
\par\noindent
\small
\begin{align}
\label{eq:sigmodelstd}
    x^{o}(t) = \sum_{j=1}^{k} c^{o}_j \delta(t-t^{o}_j),
\end{align}
\normalsize
where $\delta(t)$ is the Dirac delta function, $c^{o}_j \neq 0 \in \mathbb{C}$, and $t^{o}_j \in [0,1)$. Its Fourier transform is given by:
\par\noindent
\small
\begin{align}
\label{eq:FourierSignalModel}
    X^{o}_f = \sum_{j=1}^{k} c^{o}_j e^{-i 2 \pi f t^{o}_j} = \sum_{j=1}^{k} |c^{o}_j| a(t^{o}_j,\phi^{o}_j)_f,\; f\in\mathcal{N},
\end{align}
\normalsize
where $f \in \mathcal{N} = \{0,1,...,n-1\}$, $a(t^{o}_j,\phi^{o}_j) \in \mathbb{C}^{|\mathcal{N}|}$ is an \textit{atom} vector, with the $f$-th element given by $a(t^{o}_j,\phi^{o}_j)_f = e^{-i (2\pi f t^{o}_j - \phi^{o}_j)}$. Simply, $X^{o} = V^{o}c^{o}$, where $X^{o} \in \mathbb{C}^{n}$, $V^{o} = [a(t^{o}_1,0),...,a(t^{o}_k,0)]$, and $c^{o} = [|c^{o}_1|e^{i\phi^{o}_1}, ..., |c^{o}_k|e^{i\phi^{o}_k}]^T$. We also define the minimum separation of $x^{o}(t)$, denoted by $\bigtriangleup_t$, as the closest distance between any two different time value $t^{o}_j$'s in cyclic manner \cite{candes2014towards,tang2012csotg}, i.e.,
\par\noindent
\small
\begin{align}
    \Delta_t = \underset{t^{o}_i,t^{o}_j \in [0,1),\;i\neq j}{\min}\;\;|t^{o}_i - t^{o}_j|.
\end{align}
\normalsize

The goal here is to find $x^{o}(t)$ from the low-frequency Fourier magnitude measurements. We state the phaseless super-resolution problem with masks as follows:
\par\noindent
\small
\begin{align}
\label{prob:orig}
    & \text{Find}\;\; x(t)  \nonumber  \\
    & \text{subject to} \;\; Z[r,l]
     = \bigg| \int_{0}^{1} D_r(t) \sum_{j=1}^{k} c^{o}_j \delta(t-t^{o}_j) e^{-i 2\pi l t} dt \bigg|,\\
    & \quad\quad\quad \text{for}\; {-R \leq r \leq R}\;\text{and}\;l \in \mathcal{N}, \nonumber
\end{align}
\normalsize
where $Z[r,l]$ is the $l$-th frequency magnitude obtained by using the $r$-th mask function $D_r(t)$. Depending on the mask function $D_r(t)$, one can have different types of magnitude information. For example, if we choose $1 + e^{-i 2\pi t}$ for $D_r(t)$, we have $|X^o_l + X^o_{l+1}|$, $l \in \mathcal{N}$.

In \cite{jaganathan2016phaseless}, Jaganathan \textit{et al.} consider the case when the signal $x^{o}(t)$ is located on the grid, i.e., $t^{o}_j \in \{0,1,2,....n-1\}$. By $n$-point DFT, (\ref{prob:orig}) is equivalent to
\par\noindent
\small
\begin{align}
\label{prob:PLS_discrete}
    & \text{Find}\;\; x  \nonumber  \\
    & \text{subject to} \;\; Z[r,l]
     = | \langle f_l, D_r x \rangle |,\\
    & \quad\quad\quad \text{for}\; {-R \leq r \leq R}\;\text{and}\;l \in \mathcal{N}, \nonumber
\end{align}
\normalsize
where $x \in \mathbb{C}^{n}$ is a complex valued $k$-sparse vector, $D_r \in \mathbb{C}^{n \times n}$ is a diagonal matrix, and $f_l$ is the conjugate of the $l$-th column of the $n$ point DFT matrix. The authors in \cite{jaganathan2016phaseless} proposed the following semidefinite relaxation-based program for the phaseless super-resolution in the discrete domain by denoting $Y=xx^H$ and relaxing the rank-1 constraint on $Y$:
\par\noindent
\small
\begin{align}
\label{prob:PLS_discrete2}
    & \underset{Y}{\text{minimize}}\;\; ||Y||_1 + \lambda \text{Tr}(Y)  \nonumber  \\
    & \text{subject to} \;\; Z[r,l]^2
     = \text{Tr}( D_r^H f_l f_l^H D_r Y ),\\
    & \quad\quad\quad \text{for}\; {-R \leq r \leq R},\;l \in \mathcal{N},\;\text{and}\;Y \succeq 0,\nonumber
\end{align}
\normalsize
for some $\lambda > 0$.

This paper makes no assumption of $t^{o}_j$ being on the grid. In the next section, we propose an ANM based semidefinite relaxation of (\ref{prob:orig}) to deal with impulse functions off the grid.

\section{Phaseless Super-Resolution in the Continuous Domain}
We define the atomic norm of a vector $X \in \mathbb{C}^{|\mathcal{N}|}$ as follows:
\par\noindent
\small
\begin{align}
\label{eq:atomicNorm}
    & ||X||_{\mathcal{A}} = \inf \{ \sum_j |c_j|: X_l = \sum_j |c_j| a(t_j, \phi_j)_l,\; \substack{{t_j \in [0,1),}\\{ \phi_j \in [0,2\pi)}}\}.
\end{align}
\normalsize
We have the following new proposition for the atomic norm:
\begin{proposition}
\label{pro:sqr_atomic}
For any $X \in \mathbb{C}^{|\mathcal{N}|}$, $\mathcal{N} =\{0,1,...,n-1\}$,
\par\noindent
\small
\begin{align}
\label{eq:sqr_atomic_norm}
    ||X||_{\mathcal{A}}^2 = \inf_{u,s} \bigg\{\frac{1}{|\mathcal{N}|}\text{$s$Tr(Toep($u$))}\;:\;\begin{bmatrix} \text{Toep($u$)} & X \\ X^H & s\end{bmatrix} \succeq 0 \bigg\},
\end{align}
\normalsize
where $\text{Tr}(\cdot)$ is the trace operator, and $\text{Toep($u$)}$ is the Toeplitz matrix whose first column is $u=[u_0, u_1,...,u_{n-1}]^T$. Moreover, suppose after the Vandermonde decomposition \cite{cara1911uber,cara1911uber2,toeplitz1911zur}, $\text{Toep($u$)} = VDV^H$, where $ V = [a(t_1,0),...,a(t_r,0)]$ and $D$ is a positive diagonal matrix. Then, there exists a vector $c$ such that $X=Vc$ and $\sum_j|c_j| = ||X||_{\mathcal{A}}$.
\normalsize
\end{proposition}
Proposition \ref{pro:sqr_atomic} is similar to Proposition II.1 in \cite{tang2012csotg}; however, Proposition \ref{pro:sqr_atomic} considers the trace of $s\text{Toep($u$)}$ instead of the sum of trace of $\text{Toep($u$)}$ and $s$. Proposition \ref{pro:sqr_atomic} is essential to derive our new SDP formulation handling phaseless measurements. For readability, we place the proof of Proposition \ref{pro:sqr_atomic} in Appendix.

Motivated by Proposition \ref{pro:sqr_atomic}, we propose the following squared atomic norm minimization for the phaseless super-resolution in the continuous domain, simply \textit{phaseless ANM}:
\par\noindent
\small
\begin{align}
\label{eq:ANMPhaseless}
	& \underset{X}{\text{minimize}}\;\;  ||X||_{\mathcal{A}}^2 \nonumber\\
	& \text{subject to}\;\; a_r(X)=b_r, \; r=1,2,...,q,
\end{align}
\normalsize
where $q$ is the total number of magnitude measurements, $a_r(X)$ is the magnitude mapping function, $|\langle a_r, X \rangle |$, $a_r \in \mathbb{C}^{|\mathcal{N}|}$, and $b_r$'s are magnitude measurement results.

From Proposition \ref{pro:sqr_atomic}, we can change (\ref{eq:ANMPhaseless}) to
\par\noindent
\small
\begin{align}
\label{eq:ANM_SDP_PL1}
	& \underset{u,X,s}{\text{minimize}}\;\;  \frac{1}{|\mathcal{N}|} \text{$s$Tr(Toep($u$))} \nonumber\\
	& \text{subject to}\;\; U \triangleq \begin{bmatrix} \text{Toep($u$)} & X \\ X^{H} & $s$ \end{bmatrix} \succeq 0, \nonumber \\
    & \quad\quad\quad\quad \;\;a_r(X)=b_r, \; r=1,2,...,q,
\end{align}
\normalsize
where $u,X \in \mathbb{C}^{|\mathcal{N}|}$ and $s \in \mathbb{C}$. From the positive semidefiniteness of $U$, $s \geq 0$, and $\text{Toep($u$)} \succeq 0$. Besides, if $X_j \neq 0$, $j \in \mathcal{N}$, then $s \neq 0$ from the non-negativeness of all principal minors of $U$ \cite{meyer2000matrix}. However, because of the magnitude constraints, (\ref{eq:ANM_SDP_PL1}) is a non-convex program.

By the Schur complement lemma \cite{boyd2004convex}, $U \succeq 0$ implies $\text{$s$Toep($u$)} - XX^{H} \succeq 0$. Since $s\text{Toep($u$)} = \text{Toep($su$)}$, by defining $Q=XX^{H}$ and $u'=su$, and getting rid of the rank constraint on $Q$, we have the following SDP relaxation for the phaseless ANM:
\par\noindent
\small
\begin{align}
\label{eq:ANM_SDP_sim}
	& \underset{Q\succeq 0, u'}{\text{minimize}}\;\; \frac{1}{|\mathcal{N}|}\text{Tr(Toep($u'$))}\nonumber\\
	& \text{subject to}\;\; \text{Toep($u'$)} - Q \succeq 0, \nonumber\\
    & \quad\quad\quad\quad\;\; A_r(Q) = b^2_r, \; r =1,2,...,q,
\end{align}
\normalsize
where $A_r(Q)$ is a mapping function, $\text{Tr}(A_r Q)$. Here, $A_r = a_r a_r^H$. 

After solving (\ref{eq:ANM_SDP_sim}), we can find the optimal $\hat{Q}$ and optimal $\text{Toep($\hat{u}$)}$. Our analysis of (\ref{eq:ANM_SDP_sim}) in the following section shows that under certain conditions, $\hat{Q} = X^{o}{X^{o}}^H$. We can recover $X^{o}$ up to global phase by the eigenvalue decomposition of $\hat{Q}$. More importantly, because of the structure of $\text{Toep($\hat{u}$)} = V^{o} D {V^{o}}^H$ for some diagonal matrix $D$, we can apply any parameter estimation method such as Prony's method \cite{kay1988spectral,stoica2005spectral,blu2008sparse} or a matrix pencil method \cite{hua1990matrix,sarkar1995using} to find the time location $t^{o}_j$'s.

\section{Performance Analysis}
\label{sec:performance}
We first consider the analysis of (\ref{eq:ANM_SDP_sim}) given a rank-1 matrix $Q$. And then, we provide the analysis of (\ref{eq:ANM_SDP_sim}). Finally, we look at one scenario having magnitude measurements from a set of masks, in which (\ref{eq:ANM_SDP_sim}) provides the desired signal recovery.
\begin{theorem}
\label{thm:atomic_norm}
For a given rank-1 positive semidefinite matrix $Q=XX^H$, $X \in \mathbb{C}^{|\mathcal{N}|}$, the following optimization problem provides the squared atomic norm of $X$, i.e., $||X||_{\mathcal{A}}^2$:
\par\noindent
\small
\begin{align}
\label{eq:main}
	& \underset{u}{\text{minimize}}\;\; \frac{1}{|\mathcal{N}|}\text{Tr(Toep($u$))}\nonumber\\
	& \text{subject to}\;\; \text{Toep($u$)} - Q \succeq 0.
\end{align}
\normalsize
\end{theorem}
\begin{proof}
We can prove it by using Proposition \ref{pro:sqr_atomic}. Defining $u=u's$, where $s>0$ is a scalar. Then we can re-state the constraint as $Toep(u') - \frac{1}{s}XX^H \succeq 0$. By the Schur complement lemma, we have the optimization problem in Proposition \ref{pro:sqr_atomic}. Therefore, from Proposition \ref{pro:sqr_atomic}, the optimal value of (\ref{eq:main}) is the same as $||X||^2_{\mathcal{A}}$.
\end{proof}

\begin{corollary}
\label{thm:atomic_norm_subset}
If (\ref{eq:ANM_SDP_sim}) gives a rank-1 solution to $Q$, then (\ref{eq:ANM_SDP_sim}) minimizes the squared atomic norm of $X$ among all vectors $X$ satisfying the given constraints $a_r(X)=b_r$, $r=1,2,...,q$.
\end{corollary}
\begin{proof}
From Theorem \ref{thm:atomic_norm}, (\ref{eq:ANM_SDP_sim}) provides the minimum squared atomic norm of $X$ among all vectors $X$ satisfying the constraints $a_r(X)=b_r$, $r=1,2,...,q$.
\end{proof}

Let us consider the case when we have low-frequency Fourier magnitude measurements from a set of masks. The main difference between \cite{candes2014towards,tang2012csotg} and our setting is that we have only magnitude measurements, instead of measurements offering both phases and magnitudes.
\begin{theorem}
\label{thm:nonconvx}
Given the magnitude measurements $|X^{o}_j|$, $|X^{o}_j + X^{o}_{j+1}|$, and $|X^{o}_j - iX^{o}_{j+1}|$, $j \in \mathcal{N}=\{0,1,...,n-1\}$, (\ref{eq:ANM_SDP_sim}) provides the unique solution $Q=X^{o}{X^{o}}^H$, and $x^{o}(t)$ is uniquely obtained up to global phase if the following conditions hold: $X_j \neq 0$, $\forall j \in \mathcal{N}$, and $\Delta_t \geq 4/|\mathcal{N}|$.
\end{theorem}
\begin{proof}
Given magnitude data, $|X^{o}_j|^2$, $|X^{o}_{j+1}|^2$, $|X^{o}_j + X^{o}_{j+1}|^2$, and $|X^{o}_j - iX^{o}_{j+1}|^2$, we can find $Q_{j,j}$, $Q_{j+1,j+1}$, $Q_{j,j+1}$ and $Q_{j+1,j}$, which are the elements of the diagonal, sub-diagonal, and super-diagonal of the matrix $Q$, by simply solving linear equations on $Q$ together. From Lemma \ref{lemma:Q} in Appendix \ref{app:lemma:Q}, we can uniquely recover $Q=X^{o}{X^{o}}^H$ and $X^{o}$ up to global phase. According to Proposition \ref{pro:sqr_atomic} and Theorem \ref{thm:atomic_norm}, (\ref{eq:ANM_SDP_sim}) with $X^{o}$ is essentially the same as the optimization problem dealing with the standard ANM \cite{tang2012csotg} or TVNM \cite{candes2014towards}. Therefore,(\ref{eq:ANM_SDP_sim}) provides unique $x^{o}(t)$ up to global phase if the separation condition holds, i.e., $\Delta_t \geq 4/|\mathcal{N}|$.
\end{proof}

\section{Numerical Experiments}
\label{sec:experiment}
We compare our phaseless ANM against the standard ANM \cite{tang2012csotg} using measurements offering both phases and magnitudes, as well as against a simple algorithm which first performs the phase retrieval \cite{candes2015phase} and then applies the standard ANM \cite{tang2012csotg} to recover the impulse functions from the recovered signal using the phase retrieval. We use CVX \cite{cvx} to solve (\ref{eq:ANM_SDP_sim}).

Fig. \ref{fig:anm_sep_combine} (a) and (b) show the probability of successful recovery from the standard ANM and the phaseless ANM respectively. We conducted 50 trials for each parameter setting and measured the success rate. At each trial, we chose one time impulse $t^{o}_1$ uniformly at random in $[0,1)$, and another time impulse $t^{o}_2$ by adding the separation $\Delta_t$ to $t^{o}_1$ in the cyclic manner. We sampled the real part and imaginary part of time coefficients $c^{o}_j$'s uniformly at random in (0,1). We consider low frequencies, i.e., $\mathcal{M}=\{0,1,...,m-1\}$, where $m<n$, $\mathcal{M} \subseteq \mathcal{N}$. For a set of masks in the phaseless ANM, we use the same masks as those of Theorem \ref{thm:nonconvx} over the index set $\mathcal{M}$. The x-axis represents the separation condition $\Delta_t$ varied from $1/n$ to $11/n$, and y-axis is the number of low-frequency Fourier measurements $m$, varied from $2$ to $30$. In fact, for the phaseless ANM, the number of magnitude measurements is $3m-2$. We evaluated the recovery performance for the signal dimension $n = 32$.  We calculated the Euclidean distance between the estimated and true time locations. If the distance is less than $10^{-3}$, then we consider the estimation successful. Numerical experiments show that our phaseless ANM can find the exact time locations in the continuous domain with the same performance as the standard ANM. For large $k$, e.g., $k=10$, our method also provides the same performance as the standard ANM. We omit the simulation results in this paper due to the space limitation.
\begin{figure}[!t]
\centering
  \includegraphics[width=3.3in]{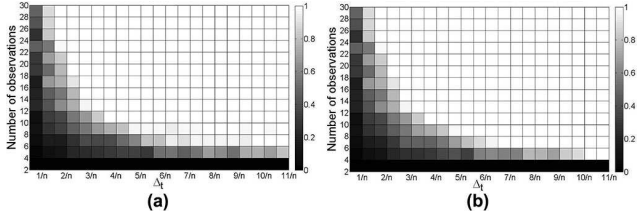}
  \caption{\small The probability $P$ of successful recovery by varying the separation condition $\Delta_t$ and the number of measurements $m$ when $n = 32$. (a) Standard ANM. (b) Phaseless ANM}
\label{fig:anm_sep_combine}
\end{figure}
\begin{figure}[!t]
\centering
  \includegraphics[width=3.3in]{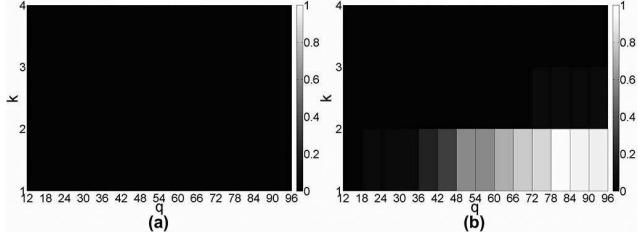}
  \caption{\small The probability $P$ of successful recovery by varying the number of magnitude measurements $q$ and sparsity $k$ when $n = 32$. (a) Phase retrieval and then standard ANM (b) Phaseless ANM}
\label{fig:PLANM_PRANM}
\end{figure}

%

One can think of a simple method conducting the phase retrieval first, and then doing the standard ANM. To compare our algorithm with this simple method, we further carried out numerical experiments by varying the number of magnitude measurements $q$ and the number of sparsity $k$ in (\ref{eq:sigmodelstd}). In this simulation, instead of using a set of masks used in Theorem \ref{thm:nonconvx}, we randomly chose a vector $a_r$ for each magnitude measurement in (\ref{eq:ANMPhaseless}). Fig. \ref{fig:PLANM_PRANM} (a) and (b) show the probability of successful recovery from the simple combining algorithm and the phaseless ANM respectively. The x-axis is the number of magnitude measurements $q$, and y-axis is the number of sparsity $k$. With randomly chosen magnitude measurements, our method outperforms the simple combining algorithm.

\section{Appendix}
\subsection{Proof of Proposition \ref{pro:sqr_atomic}}
We follow the proof of \cite[Proposition II.1]{tang2012csotg}.
\begin{proof}
\label{prf:sqr_atomic}
Let us denote the optimal value of the right hand side of (\ref{eq:sqr_atomic_norm}) by SDP($X$). In order to show $||X||_{\mathcal{A}}^2 = \text{SDP($X$)}$, we will show that (1) $||X||_{\mathcal{A}}^2 \geq \text{SDP($X$)}$ and (2) $||X||_{\mathcal{A}}^2 \leq \text{SDP($X$)}$.

The proof of (1) is easily shown by considering a feasible solution of SDP($X$). For $X=\sum_j |c_j| a(t_j, \phi_j)$, by choosing a feasible solution, $\text{Toep($u$)} = \sum_j |c_j| a(t_j,\phi_j) a(t_j,\phi_j)^H$, and $s=\sum_j |c_j|$, we have
\par\noindent
\small
\begin{align*}
    \begin{bmatrix} \text{Toep($u$)} & X \\ X^H & s \end{bmatrix} = \sum_j |c_j| \begin{bmatrix} a(t_j,\phi_j) \\ 1 \end{bmatrix}\begin{bmatrix} a(t_j,\phi_j) \\ 1 \end{bmatrix}^H  \succeq 0.
\end{align*}
\normalsize
For this feasible solution, $\frac{1}{|\mathcal{N}|} s \text{Tr(Toep($u$))} = (\sum_j |c_j|)^2$, which is $||X||_{\mathcal{A}}^2$. Thus, $\text{SDP($X$)} \leq ||X||_{\mathcal{A}}^2$.

For the proof of (2), we will show that for any $u$, $s$, and $X$, $\frac{1}{|\mathcal{N}|}s \text{Tr(Toep($u$))} \geq ||X||_{\mathcal{A}}^2$. Suppose for some $u$, $s \neq 0$, and $X$, the matrix $U$ in (\ref{eq:ANM_SDP_PL1}) is positive semidefinite. From the positive semidefinite condition, we have $\text{Toep($u$)} \succeq 0$ and $s > 0$. From the Vandermonde decomposition \cite{cara1911uber,cara1911uber2,toeplitz1911zur}, for any positive semidefinite $\text{Toep($u$)}$, we have $\text{Toep($u$)} = VDV^H$, where $V = [a(t_1,0)\;a(t_2,0),...a(t_r,0)]$, and $D$ is a diagonal matrix having $d_j$ as its $j$-th diagonal element. Since $VDV^H = \sum_{j=1}^r d_j a(t_j,0) a(t_j,0)^H$ and $||a(t_j,0)||^2_2=|\mathcal{N}|$, we have $\frac{1}{|\mathcal{N}|} \text{Tr(Toep($u$))} = \text{Tr($D$)}$. Also, from the Vandermonde decomposition and $U \succeq 0$, $X$ is in the range space of $V$; in fact, if $X$ is not in the range of $V$, we can always find a vector $z$ such that $z^H Uz < 0$. Therefore, $X=Vw =\sum_{j=1}^{r} w_j a(t_j,0) $, where $w \in \mathbb{C}^r$. By the Schur complement lemma \cite{boyd2004convex},  $U$ in (\ref{eq:ANM_SDP_PL1}) is expressed as follows:
\par\noindent
\small
\begin{align}
\label{eq:schur}
    VDV^H -\frac{1}{s}Vww^HV^H \succeq 0.
\end{align}
\normalsize
It is noteworthy that we can always find a vector $q$ such that $V^Hq = sign(w)$, where $sign(w)^H w = \sum_{j=1}^{r} |w_j|$, by choosing $q = V(V^H V)^{-1} sign(w)$. This is because $V^H$ has full row rank.
By choosing $q$ such that $V^H q = sign(w)$,  we have
\par\noindent
\small
\begin{align*}
    \text{Tr($D$)} = q^H VDV^H q \geq \frac{1}{s} q^H Vww^H V^H q = \frac{1}{s} (\sum_j |w_j|)^2,
\end{align*}
\normalsize
where the inequality is from (\ref{eq:schur}). Therefore, we have
\par\noindent
\small
\begin{align*}
    \frac{1}{|\mathcal{N}|}s\text{Tr(Toep($u$))} = s\text{Tr($D$)} \geq (\sum_j |w_j|)^2 = ||X||_{\mathcal{A}}^2.
\end{align*}
\normalsize

If $s=0$, from the sufficient and necessary condition for the positive semidefiniteness of a Hermitian matrix, all of $U$'s principal minors need to be non-negative \cite{meyer2000matrix}. Thus, $X_j = 0$, $\forall j \in \mathcal{N}$. In this case, Proposition \ref{pro:sqr_atomic} still holds.
\end{proof}


\subsection{Lemma for the positive semidefinite matrix $Q$}
\label{app:lemma:Q}
\begin{lemma}
\label{lemma:Q}
Let $Q \in \mathbb{C}^{|\mathcal{N}| \times |\mathcal{N}|}$, and $X^{o} \in \mathbb{C}^{|\mathcal{N}|}$. Suppose (1) $Q \succeq 0$, (2) $Q_{j,j} = Q^{o}_{j,j}$, $Q_{j,j+1} = Q^{o}_{j,j+1}$, and $Q_{j+1,j}=Q^{o}_{j+1,j}$, $j \in \mathcal{N}=\{0,1,...,n-1\}$, where $Q^{o}=X^{o}{X^{o}}^H$, (3) $X^{o}_j \neq 0$, $\forall j \in \mathcal{N}$. Then, $Q$ is uniquely determined as $Q=X^{o}{X^{o}}^H$.
\end{lemma}
\begin{proof}
From the fact that a Hermitian matrix is positive semidefinite if and only if all of its principal minors are non-negative \cite{meyer2000matrix}, all of $Q$'s principal minors are required to be non-negative. Let us prove our lemma by induction. When $|\mathcal{N}|=3$, the determinant of $Q$ is $-|X^{o}_1 Q_{0,2} - X^{o}_0 X^{o}_1 {X_2^{o}}^{*}|^2$, where $Q_{0,2}$ is unknown. To be $-|X^{o}_1 Q_{0,2} - X^{o}_0 X^{o}_1 {X_2^{o}}^{*}|^2 \geq 0$, $X^{o}_1 Q_{0,2} - X^{o}_0 X^{o}_1 {X_2^{o}}^{*} = 0$. Since $X^{o}_1 \neq 0$, $Q_{0,2}$ is determined uniquely as $X^{o}_0 {X_2^{o}}^{*}$. When $|\mathcal{N}|=4$, we can consider the top-left $3 \times 3$ submatrix of $Q$ and the bottom-right $3 \times 3$ submatrix of $Q$ to determine $Q_{0,2}$ and $Q_{1,3}$ respectively. And then, we can deal with $3 \times 3$ principal submatrix of $Q$ having $Q_{0,4}$ to determine $Q_{0,4}$. In the similar way, when $|\mathcal{N}|=n$, we can uniquely determine every unknown variables in $Q$. We omit the detailed explanation due to the space limitation.
\end{proof}

\small
\bibliographystyle{IEEEbib}
\bibliography{refs_PLSuper}

\end{document}